\documentclass[11pt]{article}

\usepackage{amsmath,amsthm,amssymb,amsfonts,verbatim}
\usepackage{graphicx}
\usepackage[colorlinks,
            linkcolor=blue,
            anchorcolor=blue,
            citecolor=blue
            ]{hyperref}
\usepackage[hmargin=1.2in,vmargin=1.2in]{geometry}

\title{Construction of optimal locally recoverable codes and connection with hypergraph }

\author{Chaoping Xing\thanks{School of Physical and Mathematical Sciences, Nanyang Technological University, Singapore.  This research is supported  by the National Research Foundation, Prime Minister's Office, Singapore under its Strategic Capability Research Centres Funding Initiative; and the Singapore MoE Tier 1 grants RG25/16 and RG21/18. Email: {\tt xingcp@ntu.edu.sg}.}
 \and Chen Yuan\thanks{Centrum Wiskunde \& Informatica, Amsterdam, Netherlands. Most of this work was done when the author was visiting the School of Physical and Mathematical Science, Nanyang Technological University, Singapore. Research supported in part by ERC H2020 grant No.74079 (ALGSTRONGCRYPTO). Email: {\tt Chen.Yuan@cwi.nl}}}
\date{}

\newtheorem{lemma}{Lemma}[section]
\newtheorem{theorem}[lemma]{Theorem}
\newtheorem{cor}[lemma]{Corollary}

\newtheorem{prop}[lemma]{Proposition}

\newtheorem{defn}{Definition}

\theoremstyle{remark}
\newtheorem{rmk}{Remark}
\newtheorem{conj}[lemma]{Conjecture}

\renewcommand{\epsilon}{\varepsilon}
\renewcommand{\le}{\leqslant}
\renewcommand{\ge}{\geqslant}

\newcommand{\vnote}[1]{}

\def\F{\mathbb{F}}

\def \mA {\mathcal{A}}

\def \mA {\mathcal{A}}
\def \mB {\mathcal{B}}

\def \mF {\mathcal{F}}

\def \mL {\mathcal{L}}

\def \Xi {{X^{[i]}}}

\newcommand{\Ga}{\alpha}
\newcommand{\Gb}{\beta}

\def \ba {{\bf a}}
\def \bc {{\bf c}}
\def \bb {{\bf b}}

\def \bi {{\bf 1}}
\def \bl {{\bf \ell}}

\def\be {{\bf e}}
\def\bh {{\bf h}}

\def \bu {{\bf u}}
\def \bv {{\bf v}}
\def \bo {{\bf 0}}

\begin{document}

\maketitle
\begin{abstract}
Recently, it was discovered by several authors that a $q$-ary optimal locally recoverable code, i.e., a locally recoverable code archiving the Singleton-type bound, can have length much bigger than $q+1$. This is quite different from the classical $q$-ary MDS codes where it is conjectured that the code length is upper bounded by $q+1$ (or $q+2$ for some special case). This discovery inspired some recent studies on length of an optimal locally recoverable code. It was  shown in \cite{LXY} that a $q$-ary optimal locally recoverable code is unbounded for $d=3,4$. Soon after, it was proved in \cite{GXY} that  a $q$-ary optimal locally recoverable code with distance $d$ and locality $r$ can have length $\Omega_{d,r}(q^{1 + 1/\lfloor(d-3)/2\rfloor})$. Recently, an explicit construction of $q$-ary optimal locally recoverable codes for distance $d=5,6$ was given in \cite{J18}.

\smallskip
In this paper, we further investigate construction of optimal locally recoverable codes along the line of using parity-check matrices. Inspired by classical Reed-Solomon codes and \cite{J18}, we equip parity-check matrices with the Vandermond  structure. It is turns out that  a parity-check matrix with the Vandermond  structure produces an optimal locally recoverable code must obey  certain disjoint property for subsets of $\F_q$. To our surprise,
this disjoint condition is equivalent to a well-studied problem in extremal graph theory. With the help of extremal graph theory, we succeed to improve all of the best known results in \cite{GXY} for $d\geq 7$. In addition, for $d=6$, we are able to remove the constraint required in \cite{J18} that $q$ is even.
\end{abstract}

\section{Introduction}
Motivated by applications in distributed and cloud storage systems, locally recoverable codes have been studied extensively in recent years.
Informally speaking,  a locally recoverable code (LRC for short) is a block code with an additional property called { locality}.
For a locally recoverable code $C$ of length $n$, dimension $k$ and locality $r$, it was shown in \cite{GHSY12} that the minimum distance $d(C)$ of $C$ is upper bounded by
 \begin{equation}\label{eq:x1}
 d(C)\le n-k-\left\lceil \frac kr\right\rceil+2.
 \end{equation}
The bound \eqref{eq:x1} is called the Singleton-type bound for locally recoverable codes.
A code achieving the above bound is usually called optimal.

\subsection{Known results}
Construction of optimal  locally recoverable codes, i.e., block codes archiving the bound \eqref{eq:x1} is of both theoretical interest and practical importance. This is a challenging task and has attracted great attention in the last few years. In literature, there are a few constructions available and some classes of  optimal  locally recoverable codes are known. A class of codes constructed earlier and known as pyramid codes \cite{HCL07} are shown to be codes that
are optimal.  In \cite{SRKV13},  Silberstein {\it et al}  proposed a two-level construction based on the  Gabidulin codes combined with a single parity-check $(r+1,r)$ code. Another construction \cite{TPD16} used two layers of MDS codes, a Reed-Solomon code and a special $(r+1,r)$ MDS code. A common shortcoming of these constructions relates to the size of the code alphabet which in all the papers is an exponential function of the code length, complicating the implementation. There  was an earlier construction of optimal locally recoverable codes given in \cite{PKLK12} with  alphabet  size comparable to code length. However, the construction in \cite{PKLK12} only produces  a specific value of the length $n$, i.e., $n=\left\lceil \frac kr\right\rceil(r+1)$. Thus, the rate of the code is very close to $1$. There are also some existence results given in \cite{PKLK12} and \cite{TB14} with less restriction on locality $r$. But both results require large alphabet which is an exponential function of the code length.

A recent breakthrough construction was given  in \cite{TB14}. This construction naturally generalizes Reed-Solomon construction which relies on the alphabet of cardinality comparable to the code length $n$. The idea behind the construction is very nice. The only shortcoming of this construction is restriction on locality $r$. Namely,  $r+1$ must be a divisor of either $q-1$ or $q$, or $r+1$ is equal to a product of a divisor of $q-1$ and a divisor of $q$ for certain $q$, where $q$ is the code alphabet. This construction was extended via automorphism group of rational function fields by Jin, Ma and Xing \cite{JMX17} and it turns out that there are more flexibility on locality and the code length can be $q+1$.
For particular locality such as $r=2,3,5,7,11$ or $23$, it was shown that there exist $q$-ary optimal locally recoverable codes with length up to $q+2\sqrt{q}$ via elliptic curves \cite{LMX18}. All these results are aimed at the optimal LRC with large distance.

Unlike classical  MDS codes, it is surprising to discover that the optimal LRCs can have super-linear code length in alphabet size $q$. Barg et.al, \cite{BHHMV} gave optimal LRCs by using algebraic surfaces of length $n \approx q^2$ when the distance $d=3$ and $r \le 4$. This inspired the construction of the optimal LRC with unbounded length and distance $d=3,4$ \cite{LXY}. Furthermore, it was shown in \cite{GXY} that an optimal LRC with $d\geq 5$ must have  length upper bounded in terms of alphabet size $q$. More precisely, they showed that the length of an optimal $q$-ary linear LRC with distance $d\ge 5$ and locality $r$ is upper bonded by  $O\left(dq^{3+\frac{4}{d-4}}\right)$. As for the lower bound, they presented an explicit construction of optimal LRCs with code length $\Omega_r\left(q^{1+\frac{1}{\lfloor (d-3)/2 \rfloor}}\right)$ provided that $d\leq r+2$, where $\Omega_r$ means that the implied constant depends on $r$. One can see that there is still huge gap between the lower bound and the upper bound. Following this discovery,
there are several works dedicated to constructing the maximum length of optimal LRCs. The paper \cite{J18} aimed at the optimal LRC with small distance $d=5$ or $6$. In particular, for $d=6$, the results given in \cite{J18} are obtained subject to the constraint that $q$ is even.

\subsection{Our results, comparison and a conjecture}
The main result of this paper can be summarized as follows.
\begin{theorem}\label{thm:main}
Suppose that $r\ge d-2$ and $(r+1)|n$. Then
\begin{enumerate}
\item[{\rm (i)}]  there exists an explicit construction of optimal locally recoverable codes with  length $n=q^{2-o(1)}$, minimum distance $d$ and locality $r$ for $d=7,8$;
\item[{\rm (ii)}] there exists an  explicit construction of optimal locally recoverable codes with  length $n=q^{\frac{3}{2}-o(1)}$, minimum distance $d$ and locality $r$ for $d=9,10$;
\item[{\rm (iii)}] there exist optimal locally recoverable codes with length $n=\Omega_{r,d}\left(q(q\log q)^{\frac{1}{\lfloor (d-3)/2 \rfloor}}\right)$, minimum distance $d$ and locality $r$ for $d\geq 11$; and
\item[{\rm (iv)}] there exists an   explicit construction of optimal locally recoverable code with  length $n=\Omega_{r,d}\left(q^{1+\frac{1}{\lfloor (d-3)/2 \rfloor}}\right)$, minimum distance $d$ and locality $r$ for a constant $d\geq 11$.
\end{enumerate}
\end{theorem}
%\begin{rmk}

The first three results are derived from extremal graph theory (see Section 5). The last one is derived from the random arguments (see Section 4).

The first two results  improve on the result in \cite{GXY} which only achieves $n=\Omega(q^{3/2})$ for $d=7,8$ and $n=\Omega(q^{4/3})$ for $d=9,10$.
The third one outperforms the result in \cite{GXY} by a $(\log q)^{\frac{1}{\lfloor (d-3)/2 \rfloor}}$ multiplicative factor.
In addition, for $d=6$, we are able to remove the constraint required in \cite{J18} that $q$ is even.

Although it was proved in \cite{GXY} that the length of an optimal locally recoverable code is upper bounded by $q^{3+O\left(\frac1d\right)}$, both the constructions in \cite{GXY} and this paper show from different angles that the length of an optimal locally recoverable code only achieve $q^{1+O\left(\frac1d\right)}$. Furthermore, via an upper bound from extremal graph theory, our construction in this paper can achieve at most $O\left(q^{1+\frac{2}{\lfloor (d-1)/2 \rfloor}}\right)$ (see Section 5). Thus, we make the following conjecture.

\begin{conj}  Every optimal locally recoverable code with   minimum distance $d$ and locality $r$ has length upper bounded by $q^{1+O\left(\frac1d\right)}$.
\end{conj}

\subsection{Our techniques}
For minimum distance $d\geq 7$, the only optimal locally recoverable codes with super-linear code length was given in \cite{GXY}.
In this paper, we present another construction for optimal LRCs for $d\ge 5$. Our idea comes from generalized Reed-Solomon codes where parity-check matrices have the Vandermond  structure. This  idea was already employed  in \cite{J18} for $d=5,6$.
Like in \cite{J18}, we divide a parity-check matrix  into disjoint blocks, each block with $r+1$ columns.
We require that each block of this matrix has a Vandermond matrix structure.
In order that the parity-check matrix with this structure produces an optimal locally recoverable code, elements in these blocks must satisfy certain disjoint property. In turns out that
a necessary and sufficient condition for which a parity-check matrix with this structure produces an optimal locally recoverable code is obtained in terms of certain disjoint property for subsets of $\F_q$.
This condition allows us to relate  optimality of  a locally recoverable code to a well-studied problem in extremal graph theory. With the help of extremal graph theory, we succeed to improve all of the best known results in \cite{GXY} for $d\geq 7$.

Furthermore, by a random or probabilistic argument, we show an existence result. Moreover, for constant $d$ the probabilistic method for  the existence result can be converted into a deterministic algorithm via method of conditional probabilities. Thus, we obtain an algorithmic construction in polynomial time, i.e., Theorem \ref{thm:main}(iv). The result of Theorem \ref{thm:main}(iv) matches the result given in \cite{GXY}. However, our parity-check matrix is more structured and this may lead to some other applications.
%\end{rmk}

\subsection{Organization}
The paper is organized as follows. In Section 2, we briefly introduce locally recoverable codes and some basic notations on graph theory. Section 3 presents a necessary and sufficient condition for which a Vandermond-type parity-check matrix  produces an optimal locally recoverable code in terms of certain disjoint property for subsets of $\F_q$. In Section 4, we first show an existence result via a probabilistic method. Then this probabilistic method is converted into  an algorithmic construction in polynomial time. Finally in Section 5, we show that the necessary and sufficient condition derived in Section 2 is equivalent to a central problem in extremal graph theory. By applying the known results from extremal graph theory, we obtain the desired results.

\section{Preliminaries}\label{sec:2}
\subsection{Locally recoverable codes}
Let $q$ be a prime power and $\F_q$ be the finite field with $q$ elements and denote by $[n]$ the set $\{1,2,\dots,n\}$.
In this paper, we consider linear locally recoverable codes only.  An $[n,k,d]$ linear code $C$ is a $k$-dimensional subspace of $\F_q^n$ with minimum (Hamming) distance $d$. The (Euclidean) dual code of $C$, denoted by $C^{\perp}$, is defined by $C^{\perp} = \{\bb\in\F_q^n:\; \bc\cdot\bb=0\ \mbox{for all}\ \bc\in C \}$, where $\bc\cdot\bb$ denotes the standard inner product of the two vectors $\bb$ and $\bc$.

Informally speaking, a block code is said with locality $r$ if  every coordinate of a given codeword can be recovered by accessing at most $r$ other coordinates of this codeword.
There are several equivalent definitions of locally recoverable codes.
A formal definition of a locally recoverable code with locality $r$ is given as follows.

\begin{defn}\label{def:2.1}{\rm A $q$-ary block code $C$ of length $n$ is called a locally recoverable code or locally repairable code (LRC for short) with locality $r$ if for any $i\in[n]$, there exists a subset ${R_i}\subseteq[n]\setminus\{i\}$ of size $r$ such that  for any $\bc=(c_1,\dots,c_n)\in C$, $c_i$ can be recovered by $\{c_j\}_{j\in {R_i}}$, i.e., for  any $i\in[n]$, there exists a subset ${R_i}\subseteq[n]\setminus\{i\}$ of size $r$ such that  for any  $\bu,\bv\in C$, $\bu_{{R_i}\cup\{i\}}=\bv_{{R_i}\cup\{i\}}$ if and only if  $\bu_{R_i}=\bv_{R_i}$. The set ${R_i}$ is called a recovering set of $i$.
}\end{defn}

 In literature, there are various definitions for locally recoverable code and all of them are equivalent. For example, we have the following two  definitions that are equivalent to Definition \ref{def:2.1}. For the sake of completeness, we give a proof.
 \begin{lemma}\label{lem:equiv}
 A $q$-ary code $C$ of length $n$ is a locally recoverable code if and only if one of the followings holds.
 \begin{itemize}
\item[{\rm (i)}]  For any $i\in[n]$, there exists a subset ${R_i}\subseteq[n]\setminus\{i\}$ of size $r$ such that position $i$ of every codeword $\bc\in C$ is determined by $\bc_{R_i}$, i.e, there is a function $f_i(x_1,\dots,x_r)$ (independent of $\bc$ and only dependent on $i$) such that $c_i=f_i(\bc_{R_i})$, where $\bc_{R_i}$ stands for the projection of $\bc$ at ${R_i}$.
\item[{\rm (ii)}] For any $i\in[n]$, there exists a subset ${R_i}\subseteq[n]\setminus\{i\}$ of size $r$ such that
\[C_{R_i}(i,\Ga)\cap C_{R_i}(i,\Gb)=\emptyset\]
for any $\Ga\neq\Gb\in\F_q$, where $C(i,\Ga)=\{\bc\in C:\; c_i=\Ga\}$ and $C_{R_i}(i,\Ga)$ denotes the projection of $C(i,\Ga)$ on ${R_i}$.
\end{itemize}
\end{lemma}
\begin{proof} Let $C$ be a $q$-ary code of length $n$. Assume that the condition in Definition \ref{def:2.1} is satisfied. For every $i\in[n]$, consider the subset $\{(c_i,\bc_{R_i}):\; \bc\in C\}$ of $\F_q^{r+1}$. As $\bc_{R_i}$ determines $c_i$, we can find a function $f_i(x_1,\dots,x_r)$ from $\F_q^r$ to $\F_q$ (independent of $\bc$) such that $c_i=f_i(\bc_R)$ for every $\bc\in C$. Conversely, if (i) holds, it is clear that $C$ is a locally recoverable code  with locality $r$.

Now assume that $C$ is a locally recoverable code  with locality $r$, i.e., (i) holds. Suppose that, for some $i\in[n]$ and all subsets ${R_i}\subseteq[n]\setminus\{i\}$ of size $r$, $C_{R_i}(i,\Ga)\cap C_{R_i}(i,\Gb)\neq\emptyset$ for $\Ga\neq\Gb\in\F_q$, i.e, there exist two codewords $\bu,\bv\in C$ such that $u_i=\Ga$, $v_i=\Gb$ and $\bu_{R_i}= \bv_{R_i}$. This is a contradiction to the fact that $\Ga=u_i=f_i(\bu_{R_i})=f_i(\bv_{R_i})=v_i=\Gb$. Conversely, assume that (ii) holds. we claim that $\bu_{{R_i}\cup\{i\}}=\bv_{{R_i}\cup\{i\}}$ if and only if  $\bu_{R_i}=\bv_{R_i}$. Otherwise, one would have two codewords $\bu,\bv\in C$ such that $\bu_{R_i}=\bv_{R_i}$ and $u_i\neq v_i$. This implies that $C_{R_i}(i,u_i)\cap C_{R_i}(i,v_i)$ contains $\bu_{R_i}$. This is a contradiction.
\end{proof}

The Singleton (upper) bound in \eqref{eq:x1}  is given in terms of minimum distance $d$. For convenience of this paper,  we can rewrite this bound in terms of dimension $k$.
\begin{lemma}\label{lem:2.1} Let $n,k,d,r$ be positive integers with $(r+1)|n$. If the Singleton-type bound \eqref{eq:x1} is achieved, then
\begin{equation}\label{eq:2}
n-k=\frac{n}{r+1}+d-2-\left\lfloor\frac{d-2}{r+1}\right\rfloor.
\end{equation}
Conversely, if $d-2\not\equiv r\pmod{r+1}$ and the equlity \eqref{eq:2} is satisfied, then  the Singleton-type bound \eqref{eq:x1} is achieved.
\end{lemma}
The proof is  straightforward and can be found in \cite{GXY}.
\begin{rmk}
\label{rmk:d=r+2}
{\rm If $d-2\equiv r\pmod{r+1}$, one can verify that \eqref{eq:2} implies that $r | k$. In this case, by \cite[Corollary 10]{GHSY12} one cannot achieve the Singleton-type bound \eqref{eq:x1} with equality and one must have $d \le n-k-\left\lceil \frac{k}{r}\right\rceil+1$. Therefore in this case we say an LRC attaining this latter bound as optimal.
 }
\end{rmk}
\begin{cor}\label{cor:2.2} If $r\ge d-2$, then an $[n,k,d]$ locally recoverable code with locality $r$ is optimal if
\begin{equation}\label{eq:3}
n-k-\frac{n}{r+1}=d-2.
\end{equation}
\end{cor}
\begin{proof} As $r\ge d-2$, $\left\lfloor\frac{d-2}{r+1}\right\rfloor=0$. Hence, \eqref{eq:2} and \eqref{eq:3} are equivalent.
\end{proof}
The locality of a locally recoverable code $C$ can be determined by a parity-check matrix of $C$ as follows. Assume that $(r+1)|n$.
Let  $m=\frac{n}{r+1}$ and let $D_i$ be $(n-k-m)\times(r+1)$ matrices.  Put
\begin{equation}\label{eq:H} H=\left(\begin{array}{c|c|c|c}
\bi&\bo&\cdots&\bo\\ \hline
\bo&\bi&\cdots&\bo \\ \hline
\vdots&\vdots&\ddots&\vdots \\ \hline
\bo&\bo&\cdots&\bi \\ \hline
D_1&D_2&\cdots&D_m
\end{array}
\right),
\end{equation}
where $\bi$ and $\bo$ stand for the all-one row vector and the zero row vector of length $r+1$, respectively.
Let $C$ be the code with $H$ as a parity-check matrix. Then it is clear that the dimension of $C$ is at least $k$. Furthermore, we claim that the locality of $C$ is $r$. Indeed, let $\bc=(c_1,c_2,\dots,c_n)$ be a codeword of $C$, then $\sum_{j=1+(r+1)i}^{(r+1)(i+1)} c_{j}=0$  for $0\le i\le m-1$ as $H \bc^T=\bo$.
Hence, a coordinate $c_{j}$ with $j\in\{1+(r+1)i,\dots, (r+1)(i+1)\}$ for some  $0\le i\le m-1$  can be repaired by $\bc_{R_j}$ with $R_j=\{1+(r+1)i,\dots, (r+1)(i+1)\}\setminus\{j\}$.

In conclusion, to see if a linear code $C$ with a parity-check matrix $H$ of the form \eqref{eq:H} is an optimal locally recoverable code, it is sufficient to check if the minimum distance of $C$ satisfies \eqref{eq:3} for $r\ge d-2$.
\subsection{Graphs}
A undirected graph $G$ is a pair $G=(V,E)$, where $V$ is a finite set and $E$ is a set consisting of some subsets of size $2$ of $V$. An element of $V$ is called a vertex and an element of $E$ is called an edge. A subgraph $G'$ of a graph $G$ is a graph whose vertex set and edge set are subsets of those of $G$.
We say that $G$ has a cycle $(v_1,\ldots,v_m)$ if $\{v_i,v_{i+1}\}\in E$ for $i=1,\ldots,m-1$ and $\{v_m,v_1\}\in E$.
The following Lemma \ref{lm:cycle} provides a simple but useful way to determine if $G$ contains a cycle. The proof can be found in any textbook about graph theory (see \cite{B98} for instance).
\begin{lemma}\label{lm:cycle}
An undirected graph $G$ contains a cycle if $|E|\geq |V|$.
\end{lemma}

Apart from the above  usual definition of graph, we also require some results on hypergraph in this paper. A hypergraph is a generalization of a graph in which an edge can join any number of vertices. Formally, a hypergraph $H$ is a pair $H=(X,E)$  where $X$ is a set of elements called vertices, and $E$ is a set of non-empty subsets of $X$ called hyperedges or edges. Therefore, $E$ is a subset of $2^X\setminus\{\emptyset\}$, where $2^X$ stands for the power set of $X$.
\begin{defn}[$r$-uniform Hypergraph (or $r$-hypergraph for short)]
A  hypergraph $H=(X,E)$  is called  $r$-uniform if every hyperedge in $E$ has size $r$. In other words,  every hyperedge of an $r$-uniform hypergraph connects exactly $r$ vertices.
\end{defn}
There are several ways to define cycles in a hypergraph that coincide with the definition of cycles in the usual graph. In this paper, we use the Berge cycle as the generalization of cycles in the usual graph.

\begin{defn}[Berge cycle]{\rm
A $r$-uniform hypergraph $H=(X,E)$ contains a Berge $k$-cycle $(v_1,\ldots,v_k)$ if there exist $k$ hyperedges $e_1,\ldots,e_k\in E$ such that $\{v_{i-1},v_i\}\subseteq e_i$ for $i=2,\ldots,k$  and $\{v_1,v_k\}\subseteq e_1$.}
\end{defn}

%In this paper, we always consider locally repairable codes that are linear over $\F_q$. Thus, a $q$-ary \LRC of length $n$, dimension $k$, minimum distance $d$ and locality $r$ is said to be a $q$-ary $[n,k,d]$-\LRC with locality $r$.

%For a locally repairable code $C$ of length $n$, dimension $k$ and locality $r$, it was shown in \cite{GHSY12} that the minimum distance $d(C)$ of $C$ is upper bounded by
% \begin{equation}\label{eq:x1}
% d(C)\le n-k-\left\lceil \frac kr\right\rceil+2.
% \end{equation}
%The bound \eqref{eq:x1} is called the Singleton-type bound for locally repairable codes.
%A code achieving this bound is usually called optimal locally repairable code (LRC).

\section{A criterion on minimum distance }

It follows from Corollary \ref{cor:2.2} that for $d\leq r+2$, a locally recoverable code with
parity-check matrix $H$ in \eqref{eq:H} is optimal provided that any $d-1$ columns of $H$ are linearly independent and each $D_i$ is
a $(d-2)\times (r+1)$ matrix. %In what follows, we focus on the scenario that each submatrix $D_i$ of $H$ is a Vandermond matrix.

Let $\F_q$ be a finite field and put $m=\frac{n}{r+1}$. Assume that $A_1,\ldots,A_{m}$ are subsets of $\F_q$, each of size $r+1$.
Let $A_i=\{a_{i,1},\ldots,a_{i,r+1}\}$ for $i=1,\ldots,m$.
Let $\ba_{i,j}=(a_{i,j}, a_{i,j}^2,\dots, a_{i,j}^{d-2})$ and put $D_i=(\ba_{i,1}^T,\ba_{i,2}^T,\ldots,\ba_{i,r+1}^T)$.
Thus, $D_i$ is a Vandermond-type matrix.
Let $\be_1,\ldots,\be_m$ be the standard basis of vector space $\F_q^m$, i.e., all components of $\be_i$ are $0$ except that the $i$-th component is $1$. Then, we can rewrite $H$ as follow.
\begin{equation}\label{eq:Vandermond}
H=\left(
    \begin{array}{ccccccc}
      \be_1^T & \cdots & \be_1^T & \cdots &\be_{m}^T &  \cdots  & \be_m^T \\
      \ba_{1,1}^T & \cdots &\ba_{1,r+1}^T & \cdots &\ba_{m,1}^T & \cdots & \ba_{m,r+1}^T\\
    \end{array}
  \right).
\end{equation}

We now present a sufficient and necessary condition under which any $d-1$ columns of the matrix $H$ in \eqref{eq:Vandermond}   are linearly independent.
\begin{theorem}\label{thm:set} For $d\ge 5$, then
any $d-1$ columns of $H$ defined in \eqref{eq:Vandermond} are linearly independent if and only if $|\bigcup_{i\in S}A_i|\geq r|S|+1$ for
any $S\subseteq [m]$ of size no more than $t=\lfloor\frac{d-1}{2}\rfloor$.
\end{theorem}
\begin{proof}
We first prove the ``if'' direction. Let $\bh_{i,j}$ be the $(i,j)$th  column of $H$, i.e., $\bh_{i,j}=(\be_i,\ba_{i,j})^T$ for $1\le i\le m$ and $1\le j\le r+1$.
Choose any $d-1$ columns $\{\bh_{i,j}\}_{1\le i\le m;j\in S_i}$ of $H$, where $S_i$ are subsets of $[r+1]$ satisfying $\sum_{i=1}^m|S_i|=d-1$. Let $H'$ be the $(n-k-m)\times (d-1)$ matrix consisting of these $d-1$ columns. We are going to show that $H'$ has rank $d-1$.
%Without loss of generality, we assume that these $d-1$ columns come from first $\ell$ blocks $A_1,\ldots,A_\ell$.
%That means only $S_1,\ldots,S_\ell$ are non-empty.
We assume that $S_i$ is either empty or of size at least $2$. Otherwise, the only column selected from $D_i$ with $|S_i|=1$ must be linearly independent from the rest $d-2$ columns. We can consider the linear independence of the rest $d-2$ columns instead.
Now, we assume that there are at most $t$ non-empty sets $S_i$.
%It follows that the sufficient condition can be applied to any collections of these non-empty $S_i$.
%It is then reduced to consider the linear independence of the rest $d-2$ columns.
%Let us return to the case $\ell\leq t$.  By condition, we have $|S|\geq d-1-(\ell-1)$. Otherwise
%there are at least $\ell$ repeated elements (count multiplicity) in $S$ and thus $|\bigcup_{i\in [\ell]}A_i|\leq \ell (r+1)-\ell=r\ell$. A contradiction occurs.
Let $A=\{a_{i,j}\}_{1\le i\le m;j\in S_i}$. Assume that $A=\{a_1,\ldots,a_s\}$ has $s$ distinct elements. If $s=d-1$, then by elementary row operations, one can find a $(d-1)\times (d-1)$ Vandermond submatrix of the form
$$
\left(
    \begin{array}{cccc}
      1 &  1  &  \cdots  & 1 \\
      \ba_{1}^T & \ba_{2}^T & \cdots  & \ba_{d-1}^T\\
    \end{array}
  \right)
$$
of $H'$, where $\ba_i=(a_i,a_i^2,\dots,a_i^{d-2})$. Thus, the rank of $H'$ is $d-1$.

We proceed to the case where
$s<d-1$.
By permuting the columns of $H'$, we obtain  a matrix  of  the following form:
$$
H_1=\left(
      \begin{array}{cccc|ccc}
        \be_{i_1}^T & \be_{i_2}^T & \cdots & \be_{i_s}^T & \be_{i_{s+1}}^T & \cdots & \be_{i_{d-1}}^T \\
        \ba_1^T &  \ba_2^T & \cdots & \ba_s^T & \ba_{s+1}^T & \cdots & \ba_{d-1}^T \\
      \end{array}
    \right),
$$
where $1\le i_1\le i_2\le\cdots\le i_{d-1}\le m$ and $\{a_{s+1},\ldots,a_{d-1}\}$ is a subset of $A$. Thus, $a_j$ belongs to $A_{i_j}$ for $1\le i\le d-1$.
By elementary column operations, we can erase $\ba_{s+i}^T$ since it also appears in one of  the first $s$ columns.
Hence, $H_1$ is equivalent to
$$
H_2=\left(
      \begin{array}{cccc|ccc}
        \be_{i_1}^T & \be_{i_2}^T & \cdots & \be_{i_s}^T & \be_{i_{s+1}}^T-\be_{k_{s+1}}^T & \cdots & \be_{i_{d-1}}^T-\be_{k_{d-1}}^T \\
        \ba_1^T &  \ba_2^T & \cdots & \ba_s^T & \mathbf{0}^T & \cdots & \mathbf{0}^T \\
      \end{array}
    \right),
$$
where  $\{k_{s+1},\dots,k_{d-1}\}$ is a subset of $\{i_1,\dots,i_{s}\}$. Since
$H_2$ is an upper left triangular block matrix, showing that $H_2$ is a full-rank matrix is equivalent to showing both  $(\ba_1^T,\ba_2^T,\ldots,\ba_s^T)$ and $(\be_{i_{s+1}}^T-\be_{k_{s+1}}^T, \cdots, \be_{i_{d-1}}^T-\be_{k_{d-1}}^T)$ have full rank.
Note that $(\ba_1^T,\ba_2^T,\ldots,\ba_s^T)$ is a $(d-2)\times s$ Vandermond matrix and hence it has full rank $s$.
It remains to show that $\be_{i_{s+1}}-\be_{k_{s+1}},\ldots,\be_{i_{d-1}}-\be_{k_{d-1}}$ are linearly independent. Suppose they were linearly dependent. Then
 there exist elements $\lambda_{s+1},\ldots,\lambda_{d-1}\in \F_q$ which are not all zero such that
$$
\sum_{i=s+1}^{d-1}\lambda_i (\be_{j_i}-\be_{k_i})=0.
$$
Let $P$ be the subset of $\{s+1,\ldots,d-1\}$ such that $\lambda_i\neq 0$ if and only if $i\in P$. It follows that
\begin{equation}\label{eq:linear}
\sum_{i\in P}\lambda_i (\be_{j_i}-\be_{k_i})=0.
\end{equation}
Let $U=\{j_i:i\in P\}$, $V=\{k_i:i\in P\}$ and $W=U\cup V$. As both $U$ and $V$ are subsets of $\{i\in[m]:\; |S_i|\ge 2\}$, we have $|W|\le t=\left\lfloor\frac{d-1}2\right\rfloor$.
Since $\lambda_i$ is nonzero for all $i\in P$,  every $\ell\in W$ must appear at least twice in the multiset
consisting of elements of $U$ and $V$. Otherwise, $\be_\ell$ could not be cancelled in \eqref{eq:linear}.
This implies $|W|\leq |P|$.

On the other hand,  for each $a_i\in A$, there is exactly one subset $A_{k_i}$ containing $a_i$ since
the first $s$ columns have $s$ distinct $\ba_i$.
Furthermore, let $t_i=|\{\ell\in U: a_i\in A_\ell|$. It follows that $\sum_{a_i\in A}t_i=|P|$ and $a_i$ belongs to $t_i+1$ subsets in $\{A_\ell: \ell\in W\}$. This implies
$$
\left|\bigcup_{\ell\in W}A_{\ell}\right|\leq \sum_{\ell\in W}|A_\ell|-\sum_{i=1}^st_i=(r+1)|W|-|P|.
$$
Combining with the condition $|\bigcup_{\ell\in W}A_{\ell}|\geq r|W|+1$ forces $|W|\geq |P|+1$. A contradiction occurs and we complete the proof of the ``if" direction.

%we let $A=\{a_i:i\in T\}$.
%For each $a_i\in A$, it belongs to $t_i$ subsets out of $S_j,j\in U$.
%Note that one of them is from $S_{k_i}$ and the other must be from $S_{j_i}$.
%It follows that $\sum_{a_i\in A} (t_i-1)=|T|$.
%Since $a_i$ belongs to $t_i$ subsets $A_i, i\in U$, we have
%$|\bigcup_{i\in U}A_i|\leq u(r+1)-\sum_{a_i\in A} (t_i-1)=u(r+1)-|T|$.
%However the condition says $|\bigcup_{i\in U}A_i|\geq ur+1$.
%Combining them together gives $u\geq |T|+1$. This implies there exists an element $j_i$ or $k_i$ in $U$
%appearing only once. A contradiction occurs.

We proceed to the ``only if'' direction. First, we claim that $|A_i\cap A_j|\leq 1$ for any $i\neq j$. Otherwise,  we may assume that $A_i\cap A_j$ contains two distinct elements $a_1$ and $a_2$. Thus, $H$ contains the four linearly dependent columns $(\be_i,\ba_1)^T,(\be_i,\ba_2)^T,(\be_j,\ba_1)^T$ and $(\be_j,\ba_2)^T$.

We prove the ``only if'' part by contradiction.
Without loss of generality, we assume that the first $s$ subsets $A_1,\ldots,A_s$ do not satisfy the condition, i.e. $|\bigcup_{i=1 }^s A_i|\leq sr$, where $s$ satisfies $s\le t$.
Define an undirected graph $G=([s],E)$ such that
$\{i,j\}\in E$ if and only if $A_i\cap A_j\neq \emptyset$. By inclusion-exclusion principle, we have
$$
rs\ge \left|\bigcup_{i=1 }^s A_i\right|\geq \sum_{i=1}^s |A_i|-\sum_{(i,j)\in E} 1=s(r+1)-|E|.
$$
This implies $|E|\geq s$. By Lemma \ref{lm:cycle}, there exists a cycle in this undirected graph. Without loss of generality, we may assume that $(1,\ldots,\ell)$ is a cycle, i.e., $\{i,i+1\}\in E$ for $i=1,\ldots,\ell-1$ and $\{\ell,1\}\in E$.
By the definition of $E$,  $A_{i}$ and $A_{i+1}$ contains a common element $\{a_{j_i}\}$. Then, we can pick two columns $(\be_i, \ba_{j_{i-1}})^T$\footnote{Define $\ba_{j_0}=\ba_{j_\ell}$ for simplicity.} and $(\be_i, \ba_{j_{i}})^T$ from the $i$-th block $D_i$ for $i=1,\ldots,\ell$. These $2\ell$ columns are linearly dependent since
$$
\sum_{i=1}^{\ell}\bigg((\be_i, \ba_{j_{i-1}})-(\be_i, \ba_{j_{i}})\bigg)=\sum_{i=1}^{\ell}(0, \ba_{j_{i-1}}-\ba_{j_{i}})=\mathbf{0}.
$$
The proof is completed.
\end{proof}
By Theorem \ref{thm:set}, we immediately obtain the following result.
\begin{theorem}\label{thm:opt}
If $t=\left\lfloor\frac{d-1}2\right\rfloor\ge 2$ and $(r+1)|n$, then there exists a $q$-ary  optimal linear LRC with length $n$, minimum distance $d$ and locality $r$ provided that there are $m=\frac{n}{r+1}$ sets $A_1,\ldots,A_m\subseteq \F_q$ such that
\begin{equation}\label{condition}
\begin{array}{ll}
              |A_i|=r+1 & \text{ for  $1\leq i\leq m$,} \\
              |\bigcup_{i\in S}A_i|\geq |S|r+1 & \text{ for any $S\subseteq [m]$  of size at most $t$.}
            \end{array}
\end{equation}
\end{theorem}
\begin{rmk} As we do not require that $q$ is even, the constraint required in \cite{J18}  that $q$ is even for $d=6$  can be removed.
\end{rmk}

\section{Random and algorithmic constructions}
In the previous section, we converted construction of optimal LRCs into a problem of finding subsets of $\F_q$ satisfying \eqref{condition}. In this section, we first present a random construction of subsets satisfying \eqref{condition}. In addition, we can derandomize this random  construction into a deterministic construction in polynomial time if $d$ is constant.

The case $t=2$, i.e., $d=5$ and $6$, is equivalent to the design of constant weight codes \cite{J18}.
In this section, we assume $t\geq 3$.
Since the algebraic structure is not important for the union of set. We replace $\F_q$ with $[q]$ from now on.
\begin{theorem}\label{thm:ex}
There exist $m=\left\lceil\frac{q^{1+\frac{1}{t-1}}}{2t^2(r+1)^{2+\frac{2}{t-1}}}\right\rceil$ sets $A_1,\ldots,A_m$ satisfying \eqref{condition} provided $q$ is large enough.
\end{theorem}
\begin{proof}
Let $X_i=\{x_{i,1},\ldots,x_{i,r+1}\},i=1,\ldots,2m$ be the set picked uniformly at random over all $r+1$-sized subsets of $[q]$.
Define the binary random variable $Y_S$ such that $Y_S=1$ if $|\bigcup_{i\in S}X_i|\leq |S|r$ and $0$ otherwise.
Our goal is to bound the expectation $E\left[\sum_{S\subseteq[2m],|S|\leq t}Y_S\right]$.
Without loss of generality, we may assume that $S=\{1,\ldots,a\}$ for some $1<a\leq t$.
We order the random variables in $X_i,i=1,\ldots,a$, i.e.,
$x_{1,1},\ldots,x_{1,r+1},\dots,x_{a,1},\ldots,x_{a,r+1}$.
We want to bound the probability of the event $Y_S=1$, i.e., at least $a$ elements repeated in this sequence.
Given an element $x_{i,j}$,  the probability that
$x_{i,j}\neq x_{i',j'}$ for some $x_{i',j'}$ prior to $x_{i,j}$ is at least $1-\frac{(i-1)(r+1)+j}{q}\geq 1-\frac{a(r+1)}{q}$.
Taking over all  sets of size at least $a$ in this sequence, the probability of $Y_S=1$ is at most
$$
\sum_{i=a}^{a(r+1)}{{a(r+1)}\choose {i}}\left(\frac{a(r+1)}{q}\right)^{i}\leq \sum_{i=a}^{a(r+1)}\frac{\big(a(r+1)\big)^{i}}{i!}\left(\frac{a(r+1)}{q}\right)^{i}
\leq \frac{1.1}{a!}\left(\frac{a^2(r+1)^2}{q}\right)^{a}.
$$
for $q\geq 10a^2(r+1)^2$.
It follows that
\begin{eqnarray*}
E\left[\sum_{S\subset[2m],|S|\leq t}Y_S\right]&=&\sum_{i=2}^{t}\sum_{S\subset[2m],|S|=i}\Pr[Y_S=1]\\
&\leq& \sum_{i=2}^{t}\binom{2m}{i}\frac{1.1}{i!}\left(\frac{i^2(r+1)^2}{q}\right)^{i}\leq
\sum_{i=2}^{t} 1.1(\frac{1}{i!})^2\left(\frac{2mi^2(r+1)^2}{q}\right)^{i}
\\
&\leq&\sum_{i=2}^{t} 1.1\left(\frac{1}{i!}\right)^2\left(\frac{q}{(r+1)^2}\right)^{\frac{i}{t-1}}
\leq 1.1\times 1.5\left(\frac{1}{t!}\right)^2\left(\frac{q}{(r+1)^2}\right)^{\frac{t}{t-1}}\\
&\leq&\frac{2}{4t^2}\left(\frac{q}{(r+1)^2}\right)^{\frac{t}{t-1}} \leq m.
\end{eqnarray*}
for $q\geq t^{2t}3^t(r+1)$ and $t\geq 3$.
The second inequality is due to $\binom{2m}{i}\leq \frac{(2m)^i}{i!}$ and
the third inequality is due to
$$\left(\frac{1}{i!}\right)^2\left(\frac{q}{(r+1)^2}\right)^{\frac{i}{t-1}}\geq 3\left(\frac{1}{(i-1)!}\right)^2\left(\frac{q}{(r+1)^2}\right)^{\frac{i-1}{t-1}}.$$

That means there exists $2m$ $(r+1)$-sized sets $A_1,\ldots,A_{2m}$ such that there are at most $m$ subsets $S\subseteq [2m]$ with
$|\bigcup_{i\in S}A_i|\leq |S|r$. For each of these $m$ subsets $S$,
remove one set from $A_i,i\in S$. The desired result follows as we remove at most $m$ sets.
\end{proof}

Theorem \ref{thm:ex} is an existence proof. However, if $t$ is a constant, it is possible to turn this argument into an
algorithm via the method of conditional probabilities.

\begin{theorem}
There exists a polynomial-time deterministic algorithm to find $m$ sets in Theorem \ref{thm:ex} provided that $t$ is a constant.
\end{theorem}
\begin{proof}
We follow the same notation in Theorem \ref{thm:ex}. Let $X_i=\{x_{i,1},\ldots,x_{i,r+1}\}$ be a random set of size $r+1$.
Our goal is to minimize $E[\sum_{S\subset[2m],|S|\leq t}Y_S]$ by fixing the set $X_i$ one by one.
Since
\begin{eqnarray*}
E\left[\sum_{S\subseteq[2m],|S|\leq t}Y_S\right]&=&\sum_{A\subset [q],|A|=r+1}
E\left[\sum_{S\subseteq[2m],|S|\leq t}Y_S|X_1=A\right]\Pr[X_1=A]\\
&=&\frac{1}{\binom{q}{r+1}}\sum_{A\subset [q],|A|=r+1}E\left[\sum_{S\subseteq[2m],|S|\leq t}Y_S|X_1=A\right],
\end{eqnarray*}
there exists a set $A$ such that  $E\left[\sum_{S\subseteq[2m],|S|\leq t}Y_S|X_1=A\right]\leq E\left[\sum_{S\subseteq[2m],|S|\leq t}Y_S\right]$.
If $r+1$ is a constant, we only need to enumerate all subsets of size $r+1$ in polynomial time.
However, if $r+1$ is not a constant, we enumerate $x_{1,1}\in X_1$ instead of the whole set, i.e.,
we minimize $E\left[\sum_{S\subseteq[2m],|S|\leq t}Y_S|x_{1,1}=a_{1,1}\right]$ for $a_{1,1}\in [q]$.
It remains to show how to compute this expectation.
Given a subset $S\subseteq [2m]$ of size $t$, let us show how to compute $E[Y_S|x_{1,1}=a_{1,1}]$.
Without loss of generality, we assume $S=\{1,\ldots,t\}$.
We list $t(r+1)$ random elements $x_{1,1}=a_{1,1},x_{1,2},\ldots,x_{1,r+1},\ldots,x_{t,1},\ldots,x_{t,r+1}$.
For large enough $q$, it suffices to compute $E[Y_S|x_{1,1}=a_{1,1}]$
by counting the number of sequences where there are exact $t$ repetitions.
There are $\binom{(r+1)t}{t}$ combinations of these $t$ positions. Let $R\subseteq [t]\times [r+1]$ be any set of $t$ positions.
we first remove these $t$ positions from the sequence. The remaining $tr$ positions in the sequence must have distinct elements and there are $\prod_{i=0}^{rt-1}(q-i)$ ways to pick these
$tr$ elements. Now we assign $1,\ldots,rt$ to these $rt$ positions and then determine the rest of sequence.
To obtain our final result, we multiply it by $\prod_{i=0}^{rt-1}(q-i)$.
For each $(i,j)\in R$, we enumerate all possible choices of
$x_{i,j}, (i,j)\in R$ and find out the number of combinations that there are exact $t$ repetitions in the resulting sequence.
There are at most $q^t$ ways to do the enumeration.
Then, we obtain the exact value of $E[Y_S|x_{1,1}=a_{1,1}]$.
Observe that there are at most $\sum_{i=2}^{t}\binom{n}{i}$ subsets $S$. Thus, this expectation can be computed
in polynomial time as $t$ is a constant.
We do it $r+1$ times so as to fix all elements in $X_1$.
Given $A_1,\ldots,A_k$, our goal is to find $X_{k+1}=A_{k+1}$
to minimize the expectation
$$E\left[\sum_{S\subseteq[2m],|S|\leq t}Y_S|X_1=A_1,\ldots,X_k=A_k\right]\leq E\left[\sum_{S\subseteq[2m],|S|\leq t}Y_S\right].$$
It can be done in the same way as $X_1$ is already fixed.
After we fix all these $2m$ sets, we will obtain $A_1,\ldots,A_{2m}$ with the same property as Theorem \ref{thm:ex} claims.
Then, we enumerate all $t$-sized subsets $S\subseteq[q]$ and do the same  as Theorem \ref{thm:ex} does.
The resulting subsets are the output of our algorithm. The number of these subsets is at least $m$.
Since $t$ is constant,
all this operation is done in polynomial time. The proof is completed.
\end{proof}
The following is a direct consequence of Theorem \ref{thm:set} and Theorem \ref{thm:ex}.
\begin{theorem} For  $d\ge 5$, put $t=\left\lfloor\frac{d-1}2\right\rfloor$. If $r\ge d-2$, $(r+1)|n$ and $q$ is sufficiently large, then
there exists  a $q$-ary $[n,k,d]$ optimal locally recoverable code with locality $r$ and $n\geq\frac{q^{1+\frac{1}{t-1}}}{2t^2(r+1)^{1+\frac{2}{t-1}}}$. The parity matrix of this code
has the form of \eqref{eq:Vandermond}.
Moreover, if $d$ is a constant,
there exists a deterministic algorithm running in polynomial time to construct this code.
\end{theorem}

\section{The connection with extremal graph theory}
To our surprise, it turns out that finding a collection of sets satisfying \eqref{condition} is equivalent to
constructing an $(r+1)$-uniform hypergraph  avoiding the small cycle.
The latter is one of the central problems in extremal graph theory and this problem is extremely difficult.

\begin{lemma}
There exist $m$ sets satisfying \eqref{condition} if and only if there exists an $(r+1)$-hypergraph $H=([q], E)$ with $|E|=m$ that does not have any Berge $\ell$-cycles for all $\ell\leq t$.
\end{lemma}
\begin{proof}
To see the equivalence of these two problems, we define an $(r+1)$-hypergraph as follows:
Let $H=(V,E)$ with $V=[q]$ and $E=\{A_1,\ldots,A_m\}$. It is clear that $H$ is an $(r+1)$-hypergraph.
Assume that there exists $k\leq t$ subsets $A_{i_1},\ldots, A_{i_{k}}$ does not satisfy the condition that
$|\bigcup_{j=1}^k A_{i_j}|\geq rk+1$. The same argument in Theorem \ref{thm:set} implies that there exists a cycle
$(1,2,\ldots ,\ell)$ such that $j\in A_{i_j}\cap A_{i_{j+1}}$.
That means $\{j-1,j\}\subseteq A_{i_j}$ for $j=2,\ldots,\ell$ and $\{1,\ell\}\subseteq A_{i_1}$.
By the definition of Berge cycle, the $(r+1)$-hypergraph $H$ contains this Berge $\ell$-cycle $(1,2,\ldots ,\ell)$.
On the other hand, assume that there exists a Berge $\ell$-cycle in $H$. Denote the $\ell$ edges of this cycle $A_{i_1},\ldots,A_{i_\ell}$.
The results follows since $|A_{i_j}\cap A_{i_{j+1}}|\geq 1$ for $i=1,\ldots,\ell-1$ and $|A_{i_1}\cap A_{i_\ell}|\geq 1$.
\end{proof}

The equivalence of both the problems allow us to make use of known results in this area. Let $\mF$ be a family of $r+1$-hypergraph. Denote by $ex_{r+1}(n,\mF)$ the maximum number
of edges in an $(r+1)$-hypergraph that does not contain any subgraphs in $\mF$.
 Denote by $BC_k$ the set of $k$-cycles.
Let $\mB_k=\{BC_2,\ldots,BC_k\}$.
One upper bound on $ex_{r+1}(n,\mB_t)$ is obtained by reducing this problem to an $m\times n$ bipartite graph with girth
more than $2t$ and apply the result in \cite{H2002}.
\begin{prop}[\cite{V2016}]\label{prop:upper}
$ex_{r+1}(n,\mB_t)$ is upper bounded by
\begin{enumerate}
\item[{\rm (i)}] $\frac{n}{r} (\frac{n}{r+1})^{\frac{2}{t-1}} + \frac{n}{r+1}$ if $t$ is odd,
\item[{\rm (ii)}] $\frac{n}{r(r+1)}n^{\frac{2}{t}}+\frac{n}{r+1}$ if $t$ is even.
\end{enumerate}
\end{prop}

Since these two problems are equivalent, Proposition \ref{prop:upper} gives an upper bound on the number $m$ of sets $A_i$.
For $t=3,4$, the following two propositions show that this upper bound is asymptotically tight. However, constructing
such hypergraph requires sophisticated knowledge in this area which is beyond the scope of this paper. We
summarize the results as follows.
\begin{prop}[\cite{personal}]
There exists explicit construction of $(r+1)$-hypergraph $H=([q], E)$ with $|E|=q^{2-o(1)}$ that contains no subgraph in $\mB_3$.
\end{prop}

\begin{prop}[Theorem 23 \cite{V2016}]
There exists explicit construction of $(r+1)$-hypergraph $H=([q], E)$ with $|E|=q^{\frac{3}{2}-o(1)}$ that contains no subgraph in $\mB_4$.
\end{prop}

Determining the exact value of $ex_{r+1}(n,\mB_t)$ for $r\geq 2$ and $t\geq 3$ is extremely difficult.
A major open problem in this area is whether $ex_{r+1}(n,\mB_t)=\Omega(n^{1+\frac{2}{t}})$.
A tighter lower bound for general $t$ can be obtained from $H$-free
random process \cite{BK}. The method in \cite{BK} can also be applied to hypergraph and add a
$\log$ factor above the probabilistic method in Theorem \ref{thm:ex}.
Again this technique is beyond our scope.
\begin{prop}[\cite{personal}]
$ex_{r+1}(n,\mB_t)=\Omega_{r,t}(n(n\log n)^{\frac{1}{t-1}})$.
\end{prop}

Theorem \ref{thm:main} summarizes all above results in the language of codes.

%\begin{rmk}
%In \cite{GXY}, they construct the optimal locally repairable code by filling the parity matrix $H$ column by column.
%They design an algorithm to determine the $i$-th column of $H$ so that it is linearly independent from any $d-2$ columns
%prior to this column. From their proof, the length of their codes is $\eta q^{1+\frac{1}{t-1}}$ with $(d-1)^{d-1}(r+1)^{(d-2)/2} \eta^{\lfloor (d-3)/2 \rfloor}<1$. Here, ${\lfloor (d-3)/2 \rfloor}=t-1$. We rewrite $\eta$ in terms of $r$ and $t$:
%$$
%\eta\geq (\frac{1}{2t})^{\frac{2t}{t-1}}(\frac{1}{r+1})^{\frac{t}{t-1}}.
%$$
%This implies the code length is $\frac{q^{1+\frac{1}{t-1}}}{(2t)^{\frac{2t}{t-1}}(r+1)^{\frac{t}{t-1}}}$.
%Let $t=\frac{r+1}{2}$ and our code length becomes $\frac{q^{1+\frac{2}{r-1}}}{8(r+1)^{3+\frac{4}{r-1}}}$.
%On the other hand, their code length becomes $\frac{q^{1+\frac{2}{r-1}}}{(r+1)^{3+\frac{6}{r-1}}}$.
%Our code is slightly longer than theirs for large enough $r$.
%\end{rmk}

\section*{Acknowledgement}
We sincerely thank Prof. J. Verstra{\"e}te for his linking our condition \eqref{condition} with the problem in extremal graph theory. He also provided us some references for latest results on extremal graph theory. We would also like to express our great gratitude to Profs. V. Guruswami, Q. Xiang and M. Lu for discussions and help.

\bibliographystyle{plain}
\bibliography{LRC}
\end{document}